\documentclass[12]{article}
\usepackage{latexsym,amsmath,amssymb,graphics,amscd}
\usepackage{graphics,color}
\usepackage{amsmath}
\usepackage{amsfonts}
\usepackage{amssymb}
\usepackage{amsthm}
\usepackage{latexsym}
\usepackage{graphicx}

\newcommand\bone{\mathbf{1}}
\newcommand\bd{\mathbf{d}}
\newcommand\bdout{\mathbf{d}^\mathrm{out}}
\newcommand\bdin{\mathbf{d}^\mathrm{in}}
\newcommand\bx{\mathbf{x}}
\newcommand\bu{\mathbf{u}}
\newcommand\bv{\mathbf{v}}

\newcommand\bxL{\mathbf{x}_\mathrm{L}}
\newcommand\bxR{\mathbf{x}_\mathrm{R}}

\newcommand{\RR}{\mathbb{R}}

\newtheorem{theorem}{Theorem}[section]

\theoremstyle{remark}

\title{Centrality-Friendship Paradoxes:\\
 When Our Friends Are More Important Than Us}

\author{Desmond J. Higham\thanks{Department of Mathematics and Statistics, 
University of Strathclyde, Glasgow, UK, G1 1XH, 
(\texttt{d.j.higham@strath.ac.uk}). Supported by EPSRC/RCUK
Established 
Career Fellowship EP/M00158X/1 and by EPSRC Programme Grant 
EP/P020720/1.
\textbf{EPSRC Data Statement:} not required.
}
}

\date{\today}

\begin{document}

\maketitle

\begin{abstract}
The friendship paradox 
states that, on average, our friends have more friends than we do.
In network terms, 
the average degree over the nodes can never exceed the average  
degree over the neighbours of nodes.
This effect, which is a classic example of sampling bias, has
attracted much attention 
in the social science and network science literature, 
with variations and extensions of the paradox being 
defined, tested and interpreted.
Here, we show that a version of the paradox holds 
rigorously for eigenvector centrality:
on average, our friends are more important than us.
We then consider
general matrix-function centrality, including  
Katz centrality, 
and give sufficient conditions 
for the paradox to hold.
We also discuss which results can be generalized 
to the cases of 
directed 
and 
weighted 
edges.
In this way, we add theoretical support for 
a field that has largely been evolving through empirical 
testing.
\end{abstract}

\section{Motivation}
\label{sec:mot}

Consider the graph in Figure~\ref{Fig:simp}. Imagine that the nodes
represent people and the edges represent reciprocated friendships.
Nodes $1,2,3,4,5,6,7,8$ have 
$4,1,1,1,3,2,3,1$ friends, respectively. 
So the average number of friends possessed by a node is
$16/8 = 2$.
Now look at the friends of each node.
The four friends of node 1 possess $\{1,1,1,3\}$ friends.
Similarly for nodes $2,3,4,5,6,7,8$ we find 
$\{4\}$, 
$\{4\}$,
$\{4\}$,
$\{4, 2, 3\}$,
$\{3,3\}$,
$\{2,,3,1\}$
and
$\{3\}$, respectively.
So the average number of friends possessed by a friend is
$42/16 = 2.625$, which is greater than $2$.
This effect---that \emph{our friends have more friends than we do, on average}---was identified
by Feld \cite{Feld91} and has become known as the \emph{friendship paradox}.
Feld showed that the friend-of-friend average always dominates the friend average, with
equality if and only if all individuals have the same number of friends.
The paradox is a classic example of sampling bias. In Figure~\ref{Fig:simp},
node 1 has $4$ friends and hence appears $4$ times in the friend-of-friend sum, whereas 
node $2$ only contributes its value $1$ on a single occasion; 
in general, highly connected nodes have a greater influence on the sum 

\begin{figure}
\begin{center}
\includegraphics[width=8cm]{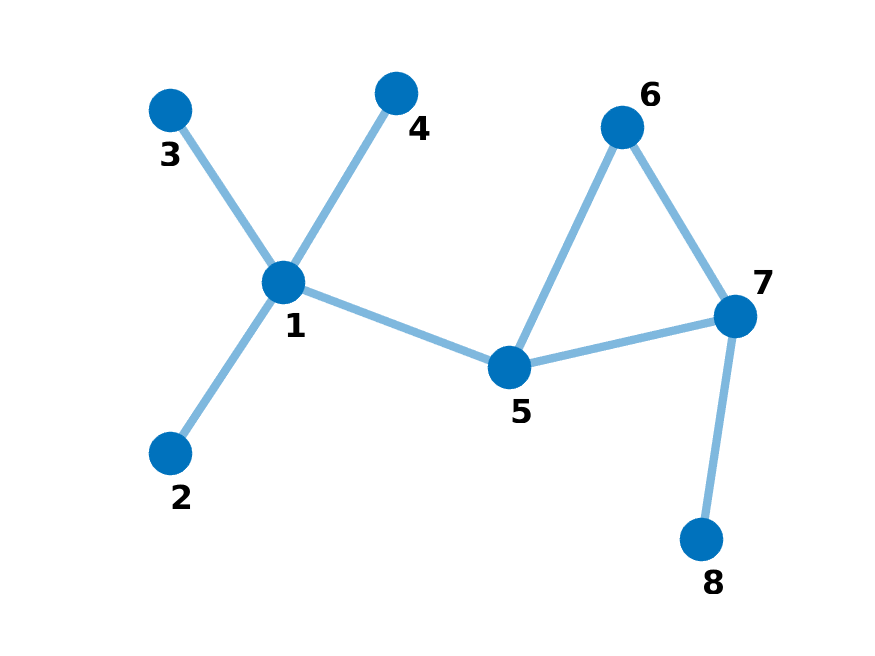}
\end{center}
  \caption{Simple undirected network with $8$ nodes.}
\label{Fig:simp}
\end{figure}
The friendship paradox has motivated much activity in
the social network literature, and is also mentioned regularly in the wider media;
see, for example, \cite{St12}.
Researchers have measured the extent to which the discrepancy holds on real networks
involving, for example, high school and university students 
\cite{Feld91,Gr14,ZJ01},
scientific coauthors
\cite{EJ14},
plants and pollinators
\cite{PMG17}
and
users of social media
\cite{Bollen2017,Hodas13icwsm,Kooti14icwsm}.
(We mention that some of these studies also 
looked at individual-level analogues, such as \lq\lq what proportion of nodes 
have fewer friends than the average over their friends?\rq\rq\
In this work we focus exclusively on 
the gobal averages used in the original reference \cite{Feld91}.)

Extending this idea, Eom and Jo \cite{EJ14} looked at the case where each node
may be quantified according to some externally derived attribute and 
studied the \emph{generalized friendship paradox}: on average, do our friends have 
more of this attribute than us?
They showed that the answer is yes for attributes that correlate positively with 
the number of friends, and found the effect to hold empirically for 
certain scientific collaboration networks in the case where the attribute 
was publication or citation count. 
Similarly, Hodas, Kooti and Lerman \cite{Hodas13icwsm} made empirical studies of 
Twitter networks and tested for a friend activity paradox (are our friends 
more active than us?) and for a 
virality paradox (do our friends spread more viral content than us?).

Our aim here is to study the generalized friendship paradox in the case where 
the attribute is
importance, as quantified by a network centrality measure.
Aside from the fact that centrality is a fundamental 
and informative nodal property \cite{Estradabook,Free,Newmanbook},
we also note that 
centrality measures are defined explicitly in terms of the network topology, and hence
there is potential to derive results that hold universally, or at least for
some well-defined classes of network.
This allows us to add further theoretical backing that complements the
recent data-driven studies mentioned above. 
Our results also alleviate the need for certain 
experiments.
For example, in \cite{Gr14} Grund 
tested the eigenvector centrality version of the generalized friendship paradox
on 
two small-scale friendship networks; Theorem~\ref{thm:ev} in Section~\ref{sec:ev} shows that
this paradox holds for \emph{all} networks.

This manuscript is organized as follows.
In Section~\ref{sec:fp} we 
set up some notation, formalize the  
friendship paradox and explain how it follows directly from the 
Cauchy--Schwarz inequality.
We also 
define the generalized friendship paradox from \cite{EJ14}, 
and show how it arises when the quantity of interest correlates with  
degree.
The new material starts in Section~\ref{sec:ev}, where we 
show that 
a paradox always holds for 
eigenvector centrality.
In Section~\ref{sec:mf} we consider other types of network centrality
based on matrix functions. 
Using a combinatorial result from \cite{LMSM83}
we show that the paradox holds for certain types of matrix function.
We also derive and interpret sufficient conditions for general 
matrix functions defined through power series with nonnegative coefficients, 
including the resolvent case corresponding to Katz centrality.
Sections~\ref{sec:ev}--\ref{sec:mf} deal with undirected, unweighted
networks. 
In Section~\ref{sec:dir} we look at directed networks, where the picture is
less straightforward.
We  
discuss various paradoxes that arise from the use of 
out-degree and in-degree, and 
give some sufficient conditions for centrality-based 
analogues.
We explain in Section~\ref{sec:weight} how all results 
extend readily to the case of nonnegatively weighted networks.
We conclude in  
Section~\ref{sec:disc} 
with an overview of the main results and an indication
of possible future lines of pursuit.

\section{Friendship Paradox and Generalized Friendship Paradox}
\label{sec:fp}

Suppose $A$ represents the adjacency matrix 
for an undirected, unweighted, network 
with $n$ nodes.
(Directed edges will be considered in 
Section~\ref{sec:dir}
and
weighted edges in
Section~\ref{sec:weight}.)
So $A \in \RR^{n \times n}$, with  $a_{ii} = 0$ and with 
$a_{ij} = a_{ji} = 1$ if nodes $i$ and $j$ are connected.
To avoid the trivial special case of an empty network, we assume at least one
edge exists.   
Letting $\bone \in \RR^{n}$ denote the vector with all components 
equal to one, we may define the degree vector 
\[
 \bd = A \bone,
\]
where $d_i$ 
gives the degree of node $i$.

We will make use of the two-norm and the one-norm, which for a vector
$\bx \in \RR^{n}$ are defined by 
\[
 \| \bx \|_2 = \sqrt{ \bx^T \bx } 
\quad \mathrm{and} \quad
 \| \bx \|_1 = \sum_{i=1}^{n} |x_i|,
\]
respectively. We note that in many cases we will be dealing with 
a nonnegatively valued vector $\bx$, whence the one-norm
reduces to the sum of the entries.

In this notation, the average degree over the nodes may be written
\[
  \frac{ \| \bd \|_1 }{n}.
\]
In the friendship paradox, we wish to compare this      
quantity 
with the average of the values that arise when we 
take each node, look at each of that node's neighbours, and record how many 
neighbours those neighbours have. 
When we do this count, each node $i$ appears as a neighbour 
$d_i$ times and each time it contributes $d_i$ neighbours, so the count totals 
$\bd^T \bd$.
The number of terms in the count is twice the number of edges,
which is $ \| \bd \|_1$. 
The overall friend-of-friend average is
therefore
\[
     \frac{ \bd^T \bd }{ \| \bd \|_1}.
\]
So 
the friendship paradox is equivalent to the inequality
\begin{equation}
     \frac{ \bd^T \bd }{ \| \bd \|_1}
 - 
  \frac{ \| \bd \|_1 }{n} 
  \ge 0.
\label{eq:fofpar}
\end{equation}

To see why 
(\ref{eq:fofpar})
is always true, we recall from the 
Cauchy--Schwarz inequality
\cite{HJ13}  
that 
for any  
$ \bu, \bv \in \RR^{n}$ we have 
\begin{equation}
     \bu^T \bv  \le \| \bu \|_2  \| \bv \|_2.
\label{eq:csineq}
\end{equation}
Taking $ \bu = \bone $, this implies 
\begin{equation}
  \| \bv \|_1 \le  \sqrt{n}  \,  \| \bv \|_2.
\label{eq:cstwo}
\end{equation}
For $\| \bv \|_1  \neq 0$, 
after squaring and rearranging we may write this
inequality in the form 
\begin{equation}
     \frac{ \bv^T \bv }{ \| \bv \|_1}
 - 
  \frac{ \| \bv \|_1 }{n} 
  \ge 0.
\label{eq:csres}
 \end{equation}
So we see that 
the friendship paradox inequality 
(\ref{eq:fofpar}) is always satisfied.

Further, 
equality holds in the 
Cauchy--Schwarz inequality
(\ref{eq:csineq})
if and only if 
$\bu $ is a multiple of $\bv$. 
So 
we have equality in 
the friendship paradox inequality 
(\ref{eq:fofpar}) 
if and only if 
the network is regular---all nodes have the same degree.

To define the generalized friendship paradox
\cite{EJ14}
a nonnegative quantity $x_i \ge 0$ is assigned to node $i$,
and
we compare the average over the nodes,
\[
  \frac{ \| \bx \|_1 }{n}, 
\]
with the average 
over neighbours of nodes,
\[
     \frac{ \bx^T \bd }{ \| \bd \|_1}.
\]
The numerator $ \bx^T \bd $ 
in the latter quantity arises because 
in the overall sum 
each node $i$ contributes 
its value $x_i$ 
a total of 
$d_i$ times. The denominator 
$ \| \bd \|_1$ 
arises because each edge is used twice.
Hence, we may say that a generalized friendship paradox 
with respect to the quantity $\bx$ arises if
\begin{equation}
     \frac{ \bd^T \bx }{ \| \bd \|_1}
 - 
  \frac{ \| \bx \|_1 }{n}
 \ge 
 0.
\label{eq:gfp}
 \end{equation}
Note that the original version 
(\ref{eq:fofpar}) corresponds to the case where $\bx$ is the degree 
vector.

Introducing the 
covariance between two vectors $\bu, \bv \in \RR^{n}$ as
\[
\mathrm{Cov}(\bu,\bv) = \frac{1}{n} \sum_{i=1}^{n} 
                         \left( x_i - \mu_x \right)
                         \left( y_i - \mu_y \right),
\]
where 
\[
 \mu_x =  \frac{1}{n} \sum_{i=1}^{n} x_i  
\quad
\mathrm{and}
\quad
 \mu_y =  \frac{1}{n} \sum_{i=1}^{n} y_i 
\]
denote the corresponding means, we may rewrite 
(\ref{eq:gfp}) as
\[
 \frac{1}{\mu_d} \mathrm{Cov}(\bd,\bx) \ge 0.
\]
Hence, as indicated in \cite{EJ14}, 
a generalized friendship paradox 
with respect to the quantity $\bx$ arises if
$\bx$ is nonnegatively correlated with degree.

Rather than considering an externally defined
attribute, as was done in the tests of
\cite{EJ14,Hodas13icwsm}, we will 
 look at circumstances where 
$\bx$ is a network 
centrality measure that quantifies the relative importance of 
each node. 
In this way we can address, 
in the same generality as the original work \cite{Feld91},
the question: 
\emph{are our friends more important than us, on average?}

\section{Eigenvector Centrality Paradox}
\label{sec:ev}

In this section we consider the case of eigenvector 
centrality \cite{B1,B2,Estradabook,Newmanbook,Vigna}. 
To make this centrality measure well-defined, 
we assume that the network is connected, and hence
the symmetric matrix $A$ is irreducible.
From Perron--Frobenius theory \cite{HJ13}, we know that
$A$ has a real, positive, dominant eigenvalue $\lambda_1$
that is equal to  $\| A \|_2$, the matrix two-norm of $A$.
The centrality measure is then given by the 
corresponding Perron--Frobenius eigenvector 
$\bx$, 
which satisfies 
$ A \bx = \lambda_1 \bx$ 
and 
has positive elements.

We have the following result.

\begin{theorem}
Given any connected network, 
the generalized friendship paradox 
inequality
(\ref{eq:gfp})
holds for eigenvector centrality, 
with equality 
if and only if the network is regular.
\label{thm:ev}
\end{theorem}

\begin{proof}
From the definition of the subordinate matrix two-norm we have
\begin{equation}
\lambda_1 = 
 \| A \|_2 \ge \| A  \frac{ \bone}{\sqrt{n}} \|_2 =  
   \| \frac{ \bd}{\sqrt{n}} \|_2.
\label{eq:anorm}
 \end{equation}
Using (\ref{eq:cstwo}), this implies 
\begin{equation}
\lambda_1
 \ge 
 \frac{1}{\sqrt{n}} \| \frac{ \bd}{\sqrt{n}} \|_1.
\label{eq:l1}
\end{equation}
Hence, 
\[
   \bd^T \bx = \bone^T A \bx = \bone^T \lambda_1 \bx
       = \lambda_1 \| \bx \|_1 \ge \frac{ \| \bd \|_1 }{n} \| \bx \|_1,
\]
and we see that 
(\ref{eq:gfp})
always holds.
Further, because 
$ \bx $
is the only eigenvector whose elements 
are all positive,  
we have equality in 
(\ref{eq:anorm})
if and only if 
$\bx $ is a multiple of $\bone$; that is, 
if and only if the network is regular.
\end{proof}

\section{Matrix Function Centrality Paradox}
\label{sec:mf}

We now move on to the case where 
$\bx$ is defined from a power series expansion
\begin{equation}
\bx = \left( c_0 I + c_1 A + c_2 A^2 + \cdots \right) \bone.
\label{eq:pws}
\end{equation}
Here, we assume that $c_k \ge 0$ for all $k$ and that these
coefficients have been chosen in such a way that the 
series converges. 
Centrality measures of this type have been studied by several 
authors, see, for example, 
\cite{BK13,BK15,E10,Estradabook,EHB12,EHSiamRev,Newmanbook,REG07,Vigna}.
They can be motivated from the combinatoric fact that 
$(A^k)_{ij}$ counts the number of distinct walks of 
length $k$ between $i$ and $j$.
Particular examples are
\begin{itemize}
 \item \emph{Katz centrality} \cite{Ka}, where $c_k = \alpha^k$. 
  Here  the real parameter $\alpha$ must be chosen such that 
  $0 < \alpha < 1/\rho(A)$, where $\rho(A)$ denotes the 
  spectral radius of $A$. In this case $\bx$ solves the 
    linear system  $(I - \alpha A) \bx = \bone$. 
  \item \emph{total centrality} \cite{BK13,BK15}, where
    $c_k = \beta^k/k!$ for some positive real parameter $\beta$.
   In this case the series converges for any $\beta$, and 
   $\bx$ may be written $\bx = \exp ( \beta A) \bone$.
   Other factorial-based coefficients have also been 
   proposed \cite{E10}. 
   \item \emph{odd and even centralities}
     based on odd and even power series, such as
    those for $\sinh$ and $\cosh$ \cite{REG07}.
\end{itemize}
We also note that 
degree centrality, on which the
the original friendship paradox is based,  
corresponds to 
$c_1 = 1$ in 
(\ref{eq:pws}) with all other coefficients equal to zero.

For the centrality measure $\bx$ in  (\ref{eq:pws}) we have 
\[
\| \bx \|_1 = 
             c_0 n + c_1 \bone^T A \bone + c_2 \bone^T A^2 \bone 
     + \cdots.
\]
  and 
\[ 
\bd^T \bx = 
 c_0 \bone^T A \bone + c_1 \bone^T A^2 \bone + 
         c_2 \bone^T A^3 \bone + \cdots.
\]
So the 
generalized friendship paradox
inequality
(\ref{eq:gfp})
may be written 
\begin{equation}
   \left(
 c_0  \bone^T A \bone + c_1 \bone^T A^2 \bone + 
         c_2 \bone^T A^3 \bone + \cdots \right)
  - 
   \left(
             c_0 n + c_1 \bone^T A \bone + c_2 \bone^T A^2 \bone 
     + \cdots \right) \frac{  \bone^T A \bone }{n}
 \ge 0.
\label{eq:pwsineq}
\end{equation}

By comparing terms in the two expansions, we arrive at the following 
sufficient condition.

\begin{theorem}
The generalized friendship paradox 
inequality
(\ref{eq:gfp})
holds for 
$\bx$ in 
(\ref{eq:pws}) if
\begin{equation}
\bone^T A^{k+1} \bone \ge 
                   \bone^T A^{k} \bone \, 
                \frac{ \bone^T A \bone}{n},
\label{eq:suff1a}
\end{equation}
for every $k \ge 1$ for which $c_k > 0$.
\label{thm:suff}
\end{theorem}

\begin{proof}
We see that
the term
on
the left hand side of 
(\ref{eq:pwsineq})
involving $c_0$
collapses to zero.
Generally, we may obtain a sufficient condition by 
asking for each individual term involving $c_k$ to be greater than 
or equal to zero, for all $k \ge 1$.
This leads to 
(\ref{eq:suff1a}).
\end{proof}

We note that the sufficient condition 
(\ref{eq:suff1a}) has a simple combinatoric interpretation:
\emph{the total number of walks of length $k+1$ 
must dominate
the product of the 
total number of walks of length $k$ and the average degree}.

To proceed we make use of the following result.
\begin{theorem}[Lagarias et al., 1983]
For any positive integers  $r$ and $s$ such that 
$r+s$ is even, we have
\begin{equation}
   \bone^T A^{r+s} \bone
   \ge 
   \frac{
   \bone^T A^{r} \bone \, 
   \bone^T A^{s} \bone } {n}.
\label{eq:lmsm}
\end{equation}
\label{thm:lag}
\end{theorem}

\begin{proof}
See \cite[Theorem~1]{LMSM83}.
\end{proof}

In words, 
Theorem~\ref{thm:lag}
says that, for $r+s$ even, \emph{the 
total number of walks of length $r+s$ dominates the 
product of the 
total number of walks of length $r$ and the  
total number of walks of length $s$, scaled by 
the number of nodes, $n$}.

This theorem allows us to deal with odd power series:
\begin{theorem}
The generalized friendship paradox
inequality
(\ref{eq:gfp})
holds for
$\bx$ in
(\ref{eq:pws}) in the case where 
$c_k = 0$ for $k$ even, 
with equality if and only if the 
network is regular.
\label{thm:odd}
\end{theorem}

\begin{proof}
First, suppose the network is regular.
Let $\mathrm{deg}$ denote the common degree, so that 
$A \bone = \mathrm{deg} \bone$. Since $\bone $ is an 
eigenvector with positive entries, it must be the 
Perron--Frobenius eigenvector,
so 
$\mathrm{deg} = \lambda_1 = \| A \|_2$.
Then $\bone^T A^{k} \bone = \bone^T \lambda_1^{k} \bone
 = n \, \lambda_1^{k}$ for all $k \ge 0$. It follows that 
for each $c_k$, 
the term on the left hand side of 
(\ref{eq:pwsineq})
collapses to zero, giving equality.

Now suppose that the network is not regular.
On the left hand side of 
(\ref{eq:pwsineq}), the 
coefficient $c_1$ receives the factor 
$\bd^T \bd - \| \bd \|_1^2 /n$. This quantity relates to the 
original friendship paradox---see 
(\ref{eq:fofpar})---and is strictly positive by the Cauchy--Schwarz 
inequality.
To deal with the remaining terms, 
it is then 
enough to show that
(\ref{eq:suff1a}) holds for odd $k>1$.
This is done by taking 
$r = k $ and $s =1$ in 
(\ref{eq:lmsm}).
\end{proof}

The next result focuses on Katz centrality.

\begin{theorem}
Consider the case where 
$c_k = \alpha^k$ in 
(\ref{eq:pws}).
For any network there exists 
a value $\alpha^\star > 0$ such that 
the 
generalized friendship paradox 
inequality
(\ref{eq:gfp})
holds for 
all
parameter values $0 < \alpha < \alpha^\star$, 
with equality if and only if the 
network is regular.
\label{thm:alphastar}
\end{theorem}
 
\begin{proof}
First, suppose the network is regular.
Recall that $\mathrm{deg}$ denotes the common degree, so that 
$\bd = \mathrm{deg} \bone$. Then
$\bx = \bone/(1-\alpha \mathrm{deg})$ is the unique solution to the  
Katz centrality equation $(I - \alpha A) \bx = \bone$ for all $0<\alpha<1/\mathrm{deg}$.
Because $\bx$ is a multiple of the degree vector, we have 
equality in 
(\ref{eq:gfp}).

Now suppose that the network is not regular, so 
$\bd$ is not a multiple of $\bone$. 
With $c_k = \alpha^k$, and $\alpha$ small, 
the left hand side of  
(\ref{eq:pwsineq}) may be expanded
as 
\[
 \alpha \left( \bd^T \bd - \frac{ \| \bd \|_1^2 }{n} \right)
+ O(\alpha^2),
\] 
and we see that the factor in parentheses is strictly positive.
\end{proof}

\section{Directed Networks}
\label{sec:dir}

In this section we consider the case of unweighted 
directed networks, so $A$ is no longer assumed to be 
symmetric.
To be concrete when discussing results, 
we imagine
that the network represents 
human-human follower relationships on  
a social media platform. So  
an edge from $i$ to $j$,
represented by $a_{ij} = 1$, 
indicates that \emph{person $i$ follows person $j$}.

We define the out-degree vector and in-degree vector by
\[
 \bdout = A \bone 
\qquad
\mathrm{and}
\qquad
 \bdin = A^T \bone,
\]
respectively.
Hence,
$d^\mathrm{out}_i$ 
counts the number of 
people that person $i$ follows, and 
$d^\mathrm{in}_j$ 
counts the number of 
people who follow person $j$. 
Note that 
\[
 \| \bdout \|_1 = \| \bdin \|_1 = \sum_{i=1}^n\sum_{j=1}^n a_{ij}.
\]

Our first observation is that the inequality 
(\ref{eq:csres}) holds for any nonzero vector 
$\bv$, with equality if and only if $\bv$ is a multiple of $\bone$. 
Hence we have 
\begin{equation}
 \frac{ {\bdout}^T \bdout }{ \| \bdout \|_1}
 - 
  \frac{ \| \bdout \|_1 }{n} 
  \ge 0
\quad
\mathrm{and}
\quad
 \frac{ {\bdin}^T \bdin }{ \| \bdin \|_1}
 - 
  \frac{ \| \bdin \|_1 }{n} 
  \ge 0.
\label{eq:outoutinin}
\end{equation}
A little care is needed when interpreting these inequalities. 
The total ${\bdout}^T \bdout$ arises if we take  
each 
person in turn, look at the people who follow them, 
and record how much following these people do. 
(In this way, each node $k$ shows up $d^\mathrm{out}_k$ times
and each time it contributes an amount $d^\mathrm{out}_k$.)
Similarly, ${\bdin}^T \bdin$ arises if we take 
each
person in turn, look at the people who they follow,
and record how many times these people are followed.
(In this way, each node $k$ shows up $d^\mathrm{in}_k$ times
and each time it contributes an amount $d^\mathrm{in}_k$.)

In words, it is always true that 
\begin{description}
 \item[i)] \emph{our followers follow 
   at least as many people as us, on average} 
  (and there is equality if and only if
  everybody follows the same number of people), and
 \item[ii)] \emph{the people we follow have 
at least as many followers 
as us, on average} (and there is equality if and only if
 everybody has the same number of followers).
\end{description}

From the discussion in Section~\ref{sec:fp}, we also see that 
the in-out/out-in analogue of 
(\ref{eq:outoutinin}) and corresponding 
statements are valid only if
$\mathrm{Cov}(\bdout,\bdin) \ge 0$.
Simple examples where 
$\mathrm{Cov}(\bdout,\bdin)$ is negative include 
the outward star graph where the only edges start at node $1$ and end at nodes $2,3,\ldots,n$,
  for which
\[
 \bdout
 = 
 \left[ 
 \begin{array}{c}
   n-1 \\ 0 \\ 0 \\ \vdots \\ 0
 \end{array}
 \right]
\quad
\mathrm{and}
\quad 
 \bdin
 = 
 \left[ 
 \begin{array}{c}
   0 \\ 1 \\ 1 \\ \vdots \\ 1
 \end{array}
 \right],
\]
and also the corresponding inward star graph. 
A strongly connected example has edges from node $1$ to nodes 
$2,3,\ldots,n$ and from node $i$ to node $i+1$ for $i = 2,3,\ldots,n-1$, plus an edge from node $n$ back to 
node $1$. 
Here, we have 
\begin{equation}
 \bdout
 = 
 \left[ 
 \begin{array}{c}
   n-1 \\ 1 \\ 1 \\ \vdots \\ 1 \\ 2
 \end{array}
 \right]
\quad
\mathrm{and}
\quad 
 \bdin
 = 
 \left[ 
 \begin{array}{c}
   1 \\ 2 \\ 2 \\ \vdots \\\vdots \\ 2
 \end{array}
 \right]. 
\label{eq:star2}
\end{equation}
In this case,  
$ {\bdout}^T \bdin = 3n-1$ and $\| \bdout \|_1 = \| \bdin \|_1 =  2n-1$, so
${\bdout}^T \bdin - \| \bdout \|_1 \| \bdin \|_1 /n = - n + 3 - 1/n$, which is negative for
$n \ge 3$. 
Hence, for such graphs it is \textbf{not true} that
\begin{description}
 \item[iii)] \emph{our followers have at least as many followers
   as us, on average}, or
 \item[iv)] \emph{the people we follow are following
at least as many people 
as us, on average}.
\end{description}

The reference \cite{Hodas13icwsm}
is unusual in that it tests the friendship paradox on 
directed networks.
The authors consider the four versions i)--iv)  
and find that the paradox holds in each case for a large social network  constructed from
Twitter data.
Our reasoning above shows that two of these versions will hold for \emph{all} directed networks. 

For an arbitrary network measure $\bx$
the relevant inequality that describes the  
out-degree version of the generalized 
friendship paradox 
(\ref{eq:gfp})
is
\begin{equation}
     \frac{ {\bdout}^T \bx }{ \| \bdout \|_1}
 - 
  \frac{ \| \bx \|_1 }{n}
 \ge 
 0.
\label{eq:gfpout}
 \end{equation}
Similarly, the in-degree version in 
\begin{equation}
     \frac{ {\bdin}^T \bx }{ \| \bdin \|_1}
 - 
  \frac{ \| \bx \|_1 }{n}
 \ge 
 0.
\label{eq:gfpin}
 \end{equation}

We now consider eigenvector centrality as our network measure.
We assume that the network is strongly connected so that $A$ is irreducible.
In this directed case, 
we have potentially distinct left and right 
Perron--Frobenius eigenvectors, which we 
denote 
$\bxL$ and $\bxR$, respectively.
Here,
$ A \bxR = \lambda_1 \bxR$ 
and
$ \bxL^T A  = \lambda_1 \bxL^T$,
with $\lambda_1 = \rho(A)$.
Both vectors $\bxL$ and $\bxR$
have positive components.

The next result characterizes two cases.

\begin{theorem}
For a strongly connected directed network 
the out-degree generalized friendship paradox  inequality 
(\ref{eq:gfpout}) 
holds for the case where
$\bx = \bxL$ if and only if
\begin{equation}
 \lambda_1 \ge \frac{ \sum_{i=1}^n \sum_{j=1}^{n} a_{ij}}{n}.
\label{eq:l1ineq}
\end{equation}
Similarly, (\ref{eq:l1ineq}) also characterizes 
the in-degree generalized friendship paradox
(\ref{eq:gfpin}) where
$\bx = \bxR$.
\label{thm:evinout}
\end{theorem} 

\begin{proof}
When $\bx = \bxL$ we have
\[
     {\bdout}^T \bx  = \bone^T A^T \bxL = 
  \bone^T \lambda_1 \bxL = \lambda_1 \| \bxL \|_1.
\]
It follows that  
(\ref{eq:gfpout}) 
reduces to 
(\ref{eq:l1ineq}).

The second statement may be proved by replacing $A$ with $A^T$.
\end{proof}
 
It is of interest to note that for our unsymmetric $A$, 
a classical result is that
$\lambda_1$ lies between the minimum out-degree or in-degree and the maximum 
out-degree or in-degree
\cite[Theorem~8.1.22]{HJ13}. However,  it is not true in general that   
$\lambda_1$ dominates the average in-degree (and hence
average out-degree).
An example is given by the strongly connected graph with
adjacency matrix
\[
A = 
\left[
 \begin{array}{ccc}
    0 & 1 & 1\\
    0 & 0 & 1 \\
    1 & 0 & 0
  \end{array}
 \right].
\]
In this case $\lambda_1$ is the real root of
$\lambda^3 - \lambda - 1$, which has the form 
\[
 \lambda_1 = \frac{-1}{3} \left( C + \frac{3}{C} \right),
\quad
\mathrm{~for~}
\quad
C = \sqrt[3]{\frac{-27 + \sqrt{621}}{2}},
\]
where $\sqrt[3]{\cdot}$ denotes the real cube root.
Here  $\lambda_1 \approx 1.3247$, which is
strictly below the average out/in degree of $4/3$.
Hence, for such networks Theorem~\ref{thm:evinout} shows that 
neither the out-degree generalized friendship paradox 
for
$\bx = \bxL$
nor the
 in-degree generalized friendship paradox
for
$\bx = \bxR$ applies. 

Matrix function based centrality measures of the form (\ref{eq:pws})
continue to make sense for directed networks.
Here, the entry 
$(A^k)_{ij}$ counts the number of distinct directed walks of 
length $k$ from $i$ to $j$.
We will focus on the Katz case, where 
\begin{equation}
 \bx = \left( I + \alpha A + \alpha^2 A + \cdots \right) \bone.
\label{eq:kundir}
\end{equation}

Let us first check how this measure correlates with in-degree.
The relevant difference 
${\bdin}^T \bx - \| \bdin \|_1 \, \| \bx \|_1 /n$ 
from (\ref{eq:gfpin}) 
then takes the form
\[
 \bone^T A \left( I + \alpha A + \alpha^2 A^2 + \cdots \right) \bone 
                          - \bone^T  \left( I + \alpha A + \alpha^2 A^2 + \cdots \right) \bone 
                             \frac{ \bone^T A \bone}{n}.
\]
In terms of powers of $\alpha$, the zeroth order term vanishes and the first order term is
\begin{equation}
 \alpha \left( \bone^T A^2 \bone - \frac{ (\bone^T A \bone)^2 }{n} \right).
\label{eq:undiff}
\end{equation}
In words, we are comparing the  
 total number of directed walks of length two with the square of the total number of directed walks of length one, scaled by the number of nodes. 
This difference can be negative---for example,
the graph that was used to give 
(\ref{eq:star2}) 
produces $\bone^T A^2 \bone - (\bone^T A \bone)^2 /n = -n + 3 - 1/n$.
Hence, the corresponding generalized friendship paradox fails on this example 
for small $\alpha$ when $n \ge 3$.

We show next that is possible to prove something positive for the alternative 
out-degree case.

\begin{theorem}
Consider the out-degree generalized friendship paradox inequality  
(\ref{eq:gfpout}) 
in the 
Katz case 
(\ref{eq:kundir}). 
For any network there exists 
a value $\alpha^\star > 0$ such that 
the 
inequality
holds for 
all
parameter values $0 < \alpha < \alpha^\star$, 
with equality if and only if the network has a common out degree.
Further, 
a sufficient condition for the inequality to hold for all  $0 < \alpha < 1/\rho(A)$ 
is
\begin{equation}
\bone^T A^T A^{k} \bone \ge 
                   \bone^T A^{k} \bone \, 
                \frac{ \bone^T A \bone}{n}. 
\label{eq:suff1}
\end{equation}
\label{thm:Kasmall}
\end{theorem}

\begin{proof}
The relevant difference is 
${\bdout}^T \bx - \| \bdout \|_1 \, \| \bx \|_1 /n$, which takes the form
\begin{equation}
 \bone^T A^T \left( I + \alpha A + \alpha^2 A^2 + \cdots \right) \bone 
                          - \bone^T \left( I + \alpha A + \alpha^2 A^2 + \cdots \right) \bone 
                             \frac{ \bone^T A \bone}{n}.
\label{eq:boutdiff}
\end{equation}
The sufficient condition 
(\ref{eq:suff1}) follows by considering powers of $\alpha$.

Now, suppose the network has a common out degree, so
$A \bone = \mathrm{deg} \bone$ for some value 
$\mathrm{deg}$.
Then $\bone$ must be the Perron--Frobenius right eigenvector, so
$A \bone = \lambda_1 \bone$, where $\lambda_1$ is the 
Perron--Frobenius eigenvalue.
In this case,
$\bone^T A^T A^k \bone = n \lambda_1^{k+1}  = \bone^T A^k \bone \bone^T A \bone$.
So we have the stated equality. 

Now, suppose the network does not have a common out degree, so
$A \bone$
is not a  multiple of $\bone$.
For small $\alpha$, the leading term in 
(\ref{eq:boutdiff}) may be written 
\[
\alpha \left(  {\bdout}^T \bdout  - \frac{ \| \bdout \|_1^2}{n} \right),
\]
which is positive by Cauchy--Schwarz. Hence a suitable $\alpha^\star$ exists.
\end{proof}

The Katz centrality measure (\ref{eq:kundir}) 
assigns to node $i$ a weighted sum of all directed walks starting from node $i$. In a message-passing context, this measure rewards nodes that are able to 
\textbf{broadcast} information effectively.
In the limit $\alpha \to 0$ this measure approaches (a shifted version of) the 
out-degree. 
 As an alternative, we could replace $A$ by $A^T$ in 
   (\ref{eq:kundir}). This measure 
assigns to node $i$ a weighted sum of all directed walks finishing at node $i$, 
 thereby rewarding nodes that are able to 
\textbf{receive} effectively.
In the limit $\alpha \to 0$ this measure approaches (a shifted version of) the 
in-degree. Analogous versions of the conclusions that follow  
(\ref{eq:undiff})
and the statement of 
Theorem~\ref{thm:Kasmall} 
are then valid with 
$A$ replaced by $A^T$ and \lq\lq in\rq\rq\ and \lq\lq out\rq\rq\ swapped.

We also note that 
part of Theorem~\ref{thm:Kasmall} extends to the 
 general
power series centrality 
measure
$ \bx = \left( c_0 I + c_1 A + c_2 A^2 + \cdots \right) \bone$---requiring 
(\ref{eq:suff1}) 
to hold 
for every $k \ge 1$ for which $c_k > 0$
serves as a sufficient condition.

\section{Weighted Networks}
\label{sec:weight}

The results in the previous sections, including 
\cite[Theorem~1]{LMSM83},
do not require the network to be unweighted.
So the conclusions extend to nonnegatively weighted networks if we are willing to 
use the 
formulations
(\ref{eq:fofpar}) and 
(\ref{eq:outoutinin})
for the friendship paradox and 
(\ref{eq:gfp}),
(\ref{eq:gfpout}) and 
(\ref{eq:gfpin}) 
 for the generalized friendship paradox, with 
the degree vectors having the same definitions:
$\bd = A \bone$, 
$\bdout = A \bone$
$\bdin = A^T \bone$.
In this extended setting, the degree vectors 
represent sum of weights rather than edge counts, and we  
note that the inequalities are invariant to positive rescaling of degree, so we may assume
without loss of generality that 
$\bd$, $\bdout$ or $\bdin$ have elements that sum to unity.
Similarly, we may assume that $\bx$ also sums to unity.
The friendship and generalized friendship paradoxes then 
apply if \lq\lq average\rq\rq\ is interpreted as \lq\lq weighted average\rq\rq.

\section{Discussion}
\label{sec:disc}

The friendship paradox has spawned a range of activity in quantitative network science, 
and it has been argued that the effect may explain reports of increasing levels of
dissatisfaction in 
online social interaction 
\cite{Bollen2017} and 
may be  
systematically distorting our perceptions and behaviours \cite{Ja}. 
It has also been shown that the paradox may be leveraged 
in order to detect the spread of information or disease, and to drive effective
interventions 
\cite{CF10, GMCCF14,Feld91,KKF18,KS18,PMG17}.
Our main result, 
Theorem~\ref{thm:ev}, shows that the paradox holds with 
the same level of generality when we consider \emph{importance},
as quantified by the classical and widely adopted eigenvector centrality measure.
Hence, our work adds further support to the argument that
this type of sampling bias is both highly relevant and ripe for 
exploitation.

It is of interest to note that the original friendship paradox is based on a purely \emph{local} 
quantity---the number of immediate neighbours.  
Theorem~\ref{thm:ev} shows that 
the effect is also present for a \emph{global} quantity that takes account of long range
interactions.
Indeed the walk-based Katz centrality measure,  
(\ref{eq:pws}) with $c_k = \alpha^k$,
interpolates between these two extremes: $\alpha \to 0$ from above
reduces to degree and $\alpha \to 1/\rho(A)$ from below
becomes eigenvector centrality
  \cite{BK15,Vigna}. Theorem~\ref{thm:alphastar} 
shows that Katz maintains the paradox for sufficiently small $\alpha$, but
it is an open question as to whether there is an undirected network for which 
the paradox fails to hold for some 
$ 0 < \alpha <  1/\rho(A)$.

We note that \cite{LMSM83} gives a concrete example of a connected network
on which the total number of walks of length three is strictly less than the product of the 
 total number of walks of length two and the average degree; so the inequality
(\ref{eq:lmsm}) is violated for $r=2$ and $s=1$. 
It follows that by taking $c_0 = c_2 = 1$ and the remaining $c_k$ sufficiently small, we can 
construct a centrality measure 
(\ref{eq:pws}) based on a power series with positive coefficients for which the generalized
friendship paradox fails to hold.
This raises the question of categorizing those power series that never give rise to such 
counterexamples. Theorem~\ref{thm:odd} shows that odd power series are one such class. 

In \cite{LMSM83} it is also stated that for any given network the 
inequality (\ref{eq:lmsm}), and hence the 
sufficiency condition 
(\ref{eq:suff1a}), holds for large enough $r+s$. This is entirely consistent with 
the eigenvector result in Theorem~\ref{thm:ev}---increasing 
$\alpha$ in Katz centrality emphasizes longer walks, and the $\alpha \to 1/\rho(A)^{-}$ 
limit corresponds to the eigenvector case \cite{BK15,Vigna}.

To the best of our knowledge, extensions of the friendship 
paradox to directed networks had 
only been studied empirically, as in 
\cite{Hodas13icwsm}.
In Section~\ref{sec:dir} we clarified that two 
of the four out/in degree versions always hold, while the other two
may fail. 
For example, when we allow for a lack of reciprocation it remains the case that  
\emph{the people we admire have more admirers than us, on average}
(something many of us first discovered at high school),
and, 
for the same reason, 
\emph{people we hate are hated by more people than us, on average}.
However,  
it is \textbf{not true} in general that 
\emph{we admire/hate people who admire/hate more people than us, on average}.

We gave in Theorem~\ref{thm:evinout} a spectral
condition 
that determines whether the relevant eigenvector centrality maintains the 
generalized paradox for directed networks. 
In this unsymmetric setting, it would be of interest to find useful classes of network 
for which the spectral condition is satisfied, and also to identify power series 
centrality measures for which a generalized friendship paradox is always guaranteed, thereby
extending the sufficiency result in Theorem~\ref{thm:Kasmall}.

\bibliography{fofrefs}

\begin{thebibliography}{10}

\bibitem{BK13}
{\sc M.~Benzi and C.~Klymko}, {\em Total communicability as a centrality
  measure}, Journal of Complex Networks, 1 (2013), pp.~124--149.

\bibitem{BK15}
{\sc M.~Benzi and C.~Klymko}, {\em On the limiting behavior of
  parameter-dependent network centrality measures}, SIAM J. Matrix Anal. Appl.,
  36 (2015), pp.~686--706.

\bibitem{Bollen2017}
{\sc J.~Bollen, B.~Gon{\c{c}}alves, I.~van~de Leemput, and G.~Ruan}, {\em The
  happiness paradox: your friends are happier than you}, EPJ Data Science, 6
  (2017), p.~4.

\bibitem{B1}
{\sc P.~Bonacich}, {\em Factoring and weighting approaches to status scores and
  clique identification}, Journal of Mathematical Sociology, 2 (1972),
  pp.~113--120.

\bibitem{B2}
\leavevmode\vrule height 2pt depth -1.6pt width 23pt, {\em Power and
  centrality: a family of measures}, American Journal of Sociology, 92 (1987),
  pp.~1170--1182.

\bibitem{CF10}
{\sc N.~A. Christakis and J.~H. Fowler}, {\em Social network sensors for early
  detection of contagious outbreaks}, PLoS ONE, 5 (2010), p.~0012948.

\bibitem{EJ14}
{\sc Y.-H. Eom and H.-H. Jo}, {\em Generalized friendship paradox in complex
  networks: The case of scientific collaboration}, Scientific Reports, 4
  (2014), p.~4603.

\bibitem{E10}
{\sc E.~Estrada}, {\em Generalized walks-based centrality measures for complex
  biological networks}, Journal of Theoretical Biology, 263 (2010),
  pp.~556--565.

\bibitem{Estradabook}
{\sc E.~Estrada}, {\em The Structure of Complex Networks}, Oxford University
  Press, Oxford, 2011.

\bibitem{EHB12}
{\sc E.~Estrada, N.~Hatano, and M.~Benzi}, {\em The physics of communicability
  in complex networks}, Physics Reports, 514 (2012), pp.~89--119.

\bibitem{EHSiamRev}
{\sc E.~Estrada and D.~J. Higham}, {\em Network propeties revealed through
  matrix functions}, SIAM Review, 52 (2010), pp.~696--671.

\bibitem{Feld91}
{\sc S.~L. Feld}, {\em Why your friends have more friends than you do},
  American Journal of Sociology, 96 (1991), pp.~1464--1477.

\bibitem{Free}
{\sc L.~C. Freeman}, {\em Centrality networks: I. conceptual clarifications},
  Social Networks, 1 (1979), pp.~215--239.

\bibitem{GMCCF14}
{\sc M.~Garcia-Herranz, E.~Moro, M.~Cebrian, N.~A. Christakis, and J.~H.
  Fowler}, {\em Using friends as sensors to detect global-scale contagious
  outbreaks}, PLoS ONE, 9 (2014), p.~0092413.

\bibitem{Gr14}
{\sc T.~U. Grund}, {\em Why your friends are more important and special than
  you think}, Sociological Science, 1 (2014), p.~128–140.

\bibitem{Hodas13icwsm}
{\sc N.~O. Hodas, F.~Kooti, and K.~Lerman}, {\em Friendship paradox redux: Your
  friends are more interesting than you}, in Proceedings of the 7Th
  International AAAI Conference On Weblogs And Social Media (ICWSM), 2013.
\newblock Honorable mention paper.

\bibitem{HJ13}
{\sc R.~A. Horn and C.~R. Johnson}, {\em Matrix Analysis}, Cambridge University
  Press, Cambridge, 2nd~ed., 2013.

\bibitem{Ja}
{\sc M.~O. Jackson}, {\em The friendship paradox and systematic biases in
  perceptions and social norms}, Journal of Political Economy,  (to appear).

\bibitem{Ka}
{\sc L.~Katz}, {\em A new index derived from sociometric data analysis},
  Psychometrika, 18 (1953), pp.~39--43.

\bibitem{Kooti14icwsm}
{\sc F.~Kooti, N.~O. Hodas, and K.~Lerman}, {\em Network weirdness: Exploring
  the origins of network paradoxes}, in International Conference on Weblogs and
  Social Media (ICWSM), Mar. 2014.

\bibitem{KKF18}
{\sc V.~Kumar, D.~Krackhardt, and S.~Feld}, {\em Network interventions based on
  inversity: {L}everaging the friendship paradox in unknown network
  structures}, Working Paper.

\bibitem{KS18}
{\sc V.~Kumar and K.~Sudhir}, {\em Can friends seed more buzz and adoption?},
  Working Paper.

\bibitem{LMSM83}
{\sc J.~C. Lagarias, J.~E. Mazo, L.~A. Shepp, and B.~McKay}, {\em An inequality
  for walks in a graph}, SIAM Review, 25 (1983), pp.~580--582.

\bibitem{Newmanbook}
{\sc M.~E.~J. Newman}, {\em Networks: an Introduction}, Oxford Univerity Press,
  Oxford, 2010.

\bibitem{PMG17}
{\sc M.~M. Pires, F.~M. Marquitti, and P.~R. Guimar{\~a}es}, {\em The
  friendship paradox in species-rich ecological networks: {I}mplications for
  conservation and monitoring}, Biological Conservation, 209 (2017),
  pp.~245--252.

\bibitem{REG07}
{\sc J.~A. Rodr\'iguez, E.~Estrada, and A.~Guti\'errez}, {\em Functional
  centrality in graphs}, Linear and Multilinear Algebra, 55 (2007),
  pp.~293--302.

\bibitem{St12}
{\sc S.~Strogatz}, {\em Friends you can count on}, The New York Times, 17th
  September (2012).

\bibitem{Vigna}
{\sc S.~Vigna}, {\em Spectral ranking}, Network Science, 4 (2016),
  pp.~433--445.

\bibitem{ZJ01}
{\sc E.~W. Zuckerman and J.~T. Jost}, {\em What makes you think you're so
  popular? {S}elf-evaluation maintenance and the subjective side of the
  {"}friendship paradox{"}}, Social Psychology Quarterly, 64 (2001),
  pp.~207--223.

\end{thebibliography}
\bibliographystyle{siam}

\end{document}